\documentclass[orivec]{llncs}

\usepackage{amsmath, amssymb, latexsym}
\usepackage{color}
\usepackage{tikz}
\usetikzlibrary{patterns}
\usetikzlibrary{decorations.markings}
\usepackage{tikz-cd}
\usetikzlibrary{arrows}

\usepackage{booktabs, tabularx}

\usepackage{stmaryrd}

\usepackage{wasysym}
 
\usepackage{hyperref}

\usepackage{tikz}
\usetikzlibrary{positioning,arrows}
\tikzset{
modal/.style={>=stealth’,shorten >=1pt,shorten <=1pt,auto,node distance=1.5cm,
semithick},
world/.style={circle,draw,minimum size=0.5cm,fill=gray!15},
point/.style={circle,draw,inner sep=0.5mm,fill=black},
reflexive above/.style={->,loop,looseness=7,in=120,out=60},
reflexive below/.style={->,loop,looseness=20,in=240,out=0},
reflexive left/.style={->,loop,looseness=7,in=150,out=210},
reflexive right/.style={->,loop,looseness=7,in=30,out=330}}

\usepackage{eucal}

\usepackage{mathtools}
\mathtoolsset{centercolon}

\newtheorem{innercustomgeneric}{\customgenericname}
\providecommand{\customgenericname}{}
\newcommand{\newcustomtheorem}[2]{%
  \newenvironment{#1}[1]
  {%
   \renewcommand\customgenericname{#2}%
   \renewcommand\theinnercustomgeneric{##1}%
   \innercustomgeneric
  }
  {\endinnercustomgeneric}
}

\newcustomtheorem{customthm}{Theorem}
\newcustomtheorem{customlemma}{Lemma}

\newcommand{\cB}{\mathcal{B}}

\newcommand{\cT}{\mathcal{T}}

\usepackage{amssymb,tikz}

\newcommand{\mysetminusD}{\hbox{\tikz{\draw[line width=0.6pt,line cap=round] (3pt,0) -- (0,6pt);}}}
\newcommand{\mysetminusT}{\mysetminusD}
\newcommand{\mysetminusS}{\hbox{\tikz{\draw[line width=0.45pt,line cap=round] (2pt,0) -- (0,4pt);}}}
\newcommand{\mysetminusSS}{\hbox{\tikz{\draw[line width=0.4pt,line cap=round] (1.5pt,0) -- (0,3pt);}}}

\newcommand{\mysetminus}{\mathbin{\mathchoice{\mysetminusD}{\mysetminusT}{\mysetminusS}{\mysetminusSS}}}
\newcommand{\draft}[1]{{\color{red}[\textsc{#1}]}}
\newcommand{\defin}[1]{\textbf{#1}}

\newcommand{\lthen}{\rightarrow}
\newcommand{\liff}{\leftrightarrow}

\newcommand{\verum}{\top}
\newcommand{\proves}{\vdash}
\newcommand{\defeq}{\coloneqq}

\newcommand{\val}[1]{[\![ #1 ]\!]}

\newcommand{\aval}[1]{[\kern-0.25em( #1 )\kern-0.25em]}

\renewcommand{\phi}{\varphi}
\newcommand{\rimp}{\Rightarrow}
\newcommand{\dimp}{\Leftrightarrow}

\newcommand{\commentout}[1]{}

\newcommand{\X}{\mathcal{X}}
\renewcommand{\L}{\mathcal{L}}
\renewcommand{\int}{\mathit{int}}
\newcommand{\cl}{\mathit{cl}}

\newcommand{\amods}{\mathrel{\kern.2em|\kern-0.2em{\approx}\kern.2em}}
\newcommand{\notamods}{\mathrel{\kern.2em|\kern-0.2em{\not\approx}\kern.2em}}

\newcommand{\fullv}[1]{#1}
\newcommand{\shortv}[1]{}

\newcommand{\Next}{\Circle}
\newcommand{\<}{\langle}
\renewcommand{\>}{\rangle}
\newcommand{\nNext}[1]{\<#1\>}

\newcommand{\pto}{\rightharpoonup}

\title{The Epistemology of Nondeterminism}
\titlerunning{The Epistemology of Nondeterminism}

\author{Adam Bjorndahl}
\institute{Carnegie Mellon University\\\email{abjorn@andrew.cmu.edu}}

\begin{document}

\maketitle

\begin{abstract}
This paper proposes new semantics for nondeterministic program execution, replacing the standard relational semantics for propositional dynamic logic (PDL). Under these new semantics, program execution is represented as fundamentally deterministic (i.e., functional), while nondeterminism emerges as an epistemic relationship between the agent and the system: intuitively, the nondeterministic outcomes of a given process are precisely those that cannot be ruled out in advance. We formalize these notions using topology and the framework of dynamic topological logic (DTL) \cite{KM05}. We show that DTL can be used to interpret the language of PDL in a manner that captures the intuition above, and moreover that continuous functions in this setting correspond exactly to deterministic processes. We also prove that certain axiomatizations of PDL remain sound and complete with respect to the corresponding classes of dynamic topological models. Finally, we extend the framework to incorporate knowledge using the machinery of subset space logic \cite{DMP96}, and show that the topological interpretation of public announcements as given in \cite{Bjorndahl18} coincides exactly with a natural interpretation of test programs.
\end{abstract}

\section{Introduction} \label{sec:int}

\emph{Propositional dynamic logic} (PDL) is a framework for reasoning about the effects of \emph{nondeterministic programs} (or, more generally, \emph{nondeterministic actions}). The standard models for PDL are relational structures interpreted as state transition diagrams: each program $\pi$ is associated with a binary relation $R_{\pi}$ on the state space, and $xR_{\pi}y$ means that state $y$ is one possible result of executing $\pi$ in $x$.\footnote{See \cite{TB15} for an overview of this branch of modal logic.}
\fullv{
\vspace{-5mm}
\begin{figure}[h]
\centering
\includegraphics[scale=0.1]{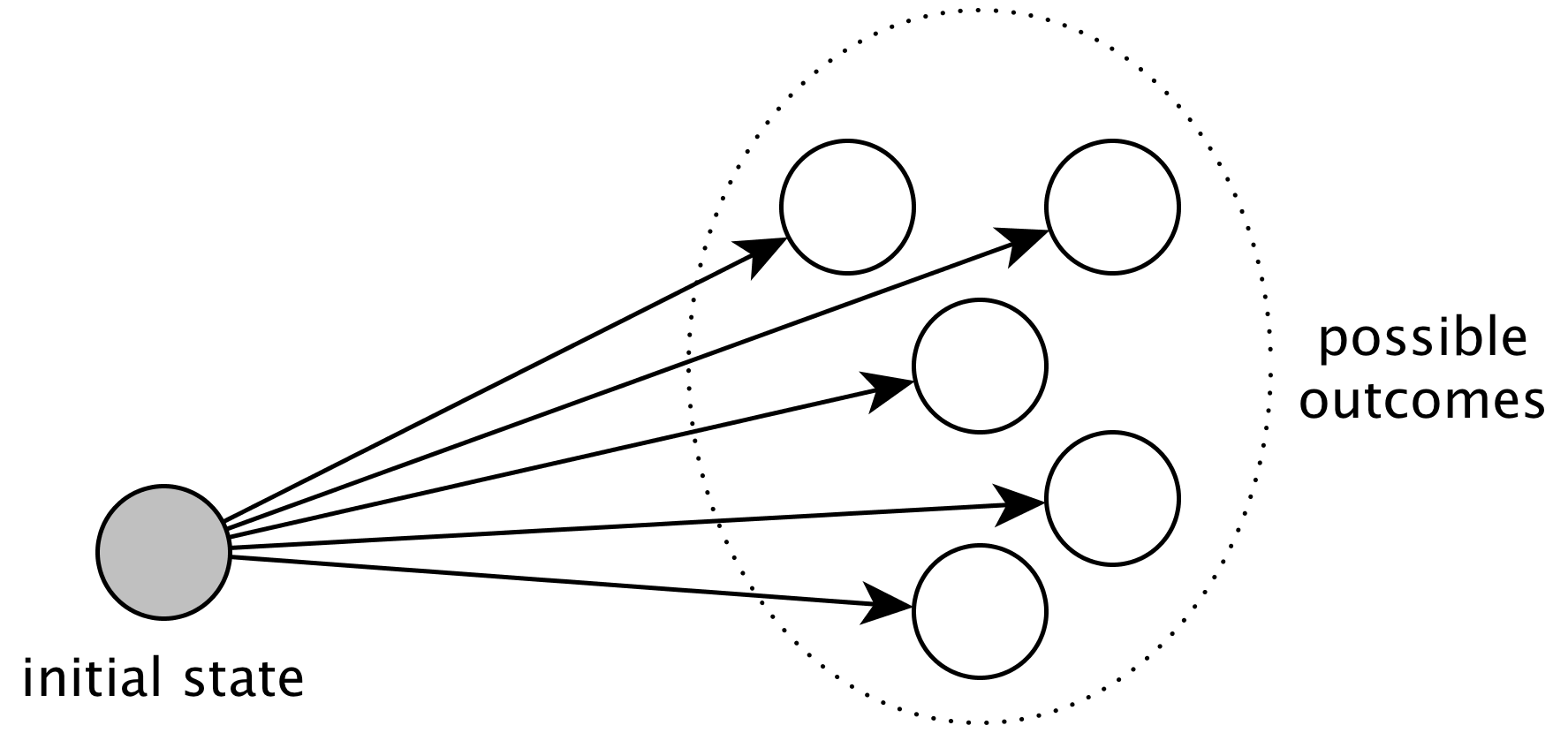}
\end{figure}
}

What is the sense of ``possibility'' at play here? This paper explores an epistemic account. The standard models for PDL treat nondeterminism as a primitive, unanalyzed notion: effectively, for each state $x$, $\pi$ is interpreted as nondeterministic at $x$ just in case $|\{y \: : \: xR_{\pi}y\}| > 1$. But one might hope for a logic that provides some insight into the \textit{nature} and \textit{source} of nondeterminism, rather than simply stipulating its existence.

We investigate a richer class of models for nondeterministic program execution which differ from the standard models in two key respects: (1) states completely determine the effects of actions, and (2) nondeterminism emerges, loosely speaking, as a kind of epistemic relationship between a given agent (or collection of agents) and the program (or action) in question. As we argue in the next section, to make this relationship precise we need structures rich enough to represent \emph{potential observations}; for this we make use of \emph{topology}.
The resulting framework is very closely related to \emph{dynamic topological logic} (DTL) as developed by Kremer and Mints \cite{KM05}; roughly speaking, we show that DTL embeds a faithful interpretation of PDL. Furthermore, we demonstrate that \emph{continuity} in this setting coincides exactly with the notion of determinism: that is, \textit{determinism is continuity in the observation topology}.

The rest of the paper is organized as follows. In Section \ref{sec:iam} we review the basics of PDL and present the intuitions that motivate the development of our new models and the importance of ``potential observations'' in the epistemic interpretation of nondeterminism. In Section \ref{sec:tdd} we motivate and review the use of topology for this purpose, and connect it to dynamic topological logic. This provides the tools we need to formalize our epistemic conception of nondeterminism and establish the correspondence between determinism and continuity mentioned above. In Section \ref{sec:amt} we transform PDL models into DTL models in a manner that preserves the truth value of all PDL formulas, and use this to prove that certain standard PDL axiomatizations are also sound and complete with respect to corresponding classes of DTL models. In Section \ref{sec:kal} we enrich our semantics using the machinery of \emph{subset space logic} \cite{DMP96} in order to reason simultaneously about both knowledge and knowability in the context of nondeterministic program execution; furthermore, we show how \emph{public announcements} \cite{Plaza07}, appropriately generalized to the topological setting \cite{Bjorndahl18}, can be captured using \emph{test programs}. Section \ref{sec:con} concludes with a brief discussion of ongoing work. Proofs and other details are collected in Appendix \ref{app:prf}.

\section{Review and Motivation} \label{sec:iam}

Fix a countable set of \emph{primitive propositions} $\textsc{prop}$ and a countable set of \emph{programs} $\Pi$. The language of PDL, denoted $\L_{PDL}$, is given by
$$\phi ::= p \, | \, \lnot \phi \, | \, \phi \land \psi \, | \, \nNext{\pi} \phi,$$
where $p \in \textsc{prop}$, $\pi \in \Pi$, and $\nNext{\pi} \phi$ is read, ``after some execution of $\pi$, $\phi$ is true''. Often $\Pi$ is constructed from a set of more ``basic'' programs by closing under certain operations, but for the moment we will take it for granted as a structureless set. A \defin{(standard) PDL model} is a relational frame $(X, (R_{\pi})_{\pi \in \Pi})$ together with a valuation $v: \textsc{prop} \to 2^{X}$; Boolean formulas are interpreted in the usual way, while $\nNext{\pi}$ is interpreted as existential quantification over the $R_{\pi}$-accessible states:
\begin{eqnarray*}
x \models p & \textrm{ iff } & x \in v(p)\\
x \models \lnot \phi & \textrm{ iff } & x \not\models \phi\\
x \models \phi \land \psi & \textrm{ iff } & x \models \phi \textrm{ and } x \models \psi\\
x \models \nNext{\pi} \phi & \textrm{ iff } & (\exists y)(x R_\pi y \textrm{ and } y \models \phi).
\end{eqnarray*}
Thus, $\nNext{\pi}\phi$ is true at a state $x$ just in case some possible execution of $\pi$ at $x$ results in a $\phi$-state.

Following standard conventions, we write $[\pi]$ as an abbreviation for $\lnot \nNext{\pi} \lnot$, so we have
\begin{eqnarray*}
x \models [\pi] \phi & \textrm{ iff } & (\forall y)(x R_\pi y \textrm{ implies } y \models \phi).
\end{eqnarray*}
We also treat $R_{\pi}$ as a set-valued function when convenient, with $R_{\pi}(x) = \{y \in X \: : \: xR_{\pi}y\}$.

It is easy to adjust the standard models for PDL so that each state completely determines the outcome of each action: simply replace the relations $R_{\pi}$ with functions $f_{\pi}: X \to X$. To emphasize this shift we introduce modalities $\Next_\pi$ into the language, reading $\Next_\pi \phi$ as ``after execution of $\pi$, $\phi$ holds''; these modalities are interpreted using the functions $f_\pi$ in the natural way:
\begin{eqnarray*}
x \models \Next_\pi \phi & \textrm{ iff } & f_\pi(x) \models \phi.
\end{eqnarray*}

Perhaps the most direct attempt to formalize nondeterminism as an epistemic notion in this setting is to interpret the ``nondeterministic outcomes'' of $\pi$ to be precisely those outcomes that the agent considers possible.
\fullv{
\begin{figure}[h]
\centering
\includegraphics[scale=0.1]{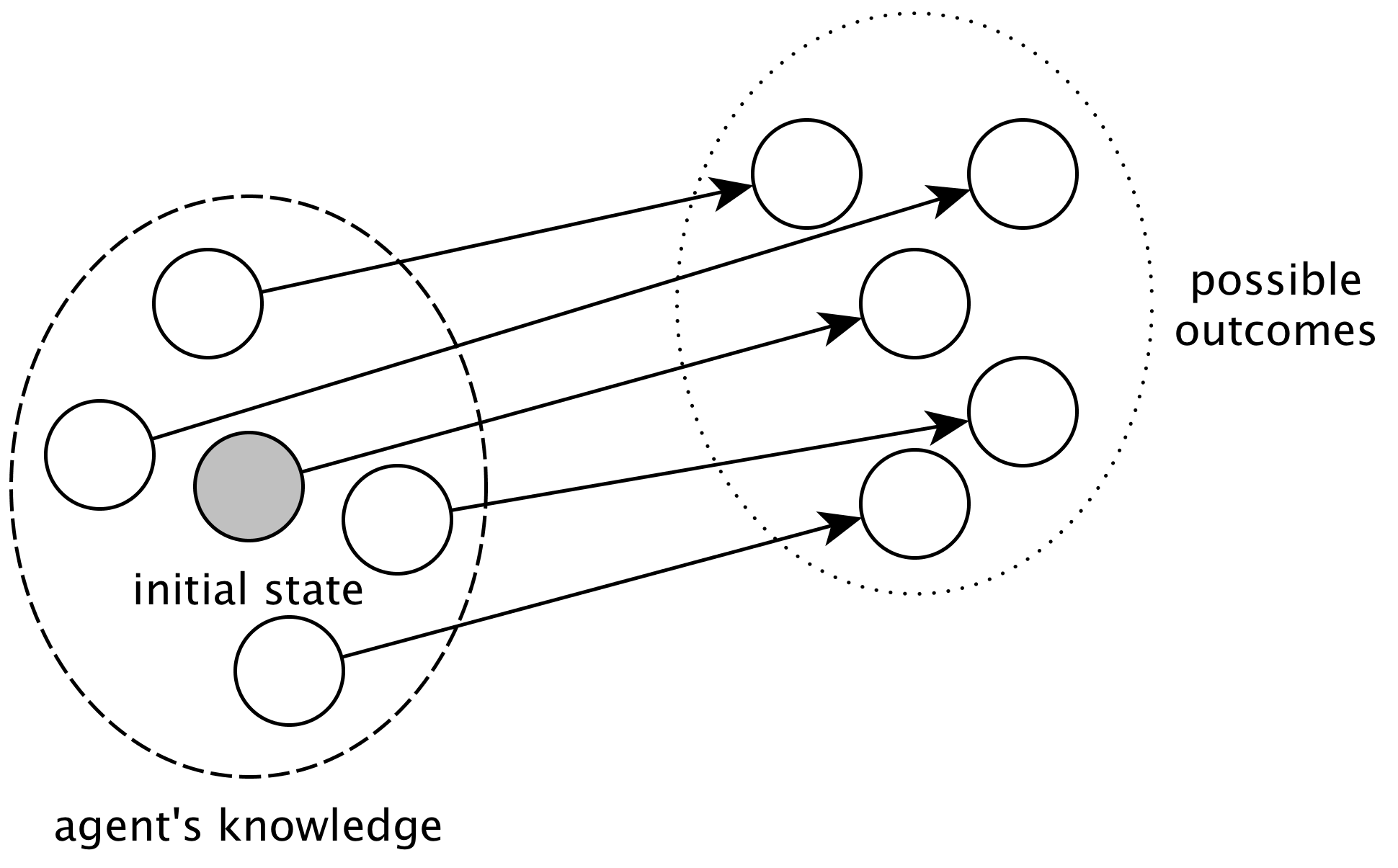}
\end{figure}
}

Somewhat more formally, supposing we have access to a knowledge modality $K$ with corresponding dual $\hat{K}$ (so $\hat{K}\phi$ is read ``the agent considers $\phi$ possible''), we might define
\begin{eqnarray*}
\nNext{\pi} \phi & \equiv & \hat{K} \Next_{\pi} \phi.
\end{eqnarray*}

Crucially, however, this seems to miss the essence of nondeterminism. For instance, according to this definition, when the agent in question happens to be very uncertain about $\pi$, we are forced to interpret $\pi$ as having a great many possible nondeterministic outcomes. But there is a clear conceptual distinction between those outcomes of $\pi$ that are possible as far as some agent knows---perhaps an agent with very poor information---as opposed to those outcomes that would remain possible \textit{even with good information}. And it seems to be the latter concept that aligns more closely with our intuitions regarding nondeterminism.

For a simple example, imagine running a random number generator. This seems like a canonical example of a nondeterministic process. Note that what is important here is not merely that you do not, in fact, know what number will be generated in advance, but also that you are unable \textit{in principle} to determine this in advance.\footnote{To be sure, if you had access to a more advanced set of tools than are standardly available, perhaps you \textit{could} make such a determination. And in this case, thinking of the random number generator as a nondeterministic process loses much of its intuitive appeal. Indeed, \textit{any} nondeterministic process whatsoever might be viewed as deterministic relative to a sufficiently powerful set of tools (e.g., from God's perspective). Thus, nondeterminism can be naturally construed as a \textit{relative} notion that depends on a fixed background set of ``feasible measurements''. We make this precise below.} By contrast, imagine running a program that queries a given database and prints the result; we would not want to call this program nondeterministic even if you happened to be ignorant about the contents of the database.

This is a distinction we want to respect. The relevant epistemic notion, then, is not what any given agent currently happens to know, but what they \textit{could come to know}. This is where topology comes in: the notion of ``potential knowledge'' or ``knowability'' is naturally represented in topological spaces.

\section{Topology, Dynamics, and Determinism} \label{sec:tdd}

\subsection{Topological Spaces and Models}

A \defin{topological space} is a set $X$ together with a collection $\cT \subseteq 2^{X}$ of subsets of $X$ such that $\emptyset, X \in \cT$ and $\cT$ is closed under unions and finite intersections. Elements of $\cT$ are called \emph{open} and $\cT$ is called the \emph{topology}.

There are various intuitions that help to make sense of this definition, most of which tap into the notion of topology as the mathematics of physical space and proximity.\footnote{For a standard introduction to topological notions, we refer the reader to \cite{munkres}.} Here, though, we focus instead on epistemic intuitions, through which topology is naturally interpreted as a formalization of evidence and potential observations. In fact, these two intuitions overlap in cases where the relevant observations are measurements about locations in space.

Informally, if we think of $X$ as a set of possible worlds encoding the potential uncertainties one may have, then we can think of open sets $U \in \cT$ as the results of measurements or observations. More precisely, we can understand $U$ to represent the observation that rules out precisely those worlds $x \notin U$. On this view, each $U \in \cT$ corresponds to a possible state of knowledge, and the topology $\cT$ itself can be conceptualized as the set of available observations.%
\footnote{Suppose, for a simple example, that you measure your height and obtain a reading of $180 \pm 2$ cm. If we represent the space of your possible heights using the positive real numbers, $\mathbb{R}^{+}$, then it is natural to identify this measurement with the open interval $(178, 182)$. And with this measurement in hand, you can safely deduce that you are, for instance, less than 183 cm tall, while remaining uncertain about whether you are, say, taller than 179 cm.}

A core notion in topology is that of the \emph{interior} of a set $A \subseteq X$, defined by:
$$\int(A) = \{x \in A \: : \: (\exists U \in \cT)(x \in U \subseteq A)\}.$$
The interior of $A$ therefore consists of those points $x$ that are ``robustly'' inside $A$, in the sense that there is some ``witness'' $U \in \cT$ to $x$'s membership in $A$. When we interpret elements of $\cT$ as the results of possible measurements, the notion of interior takes on a natural epistemic interpretation: $x$ lies in the interior of $A$ just in case there is \textit{some} measurement one could potentially take that would entail $A$. In other words, the worlds in the interior of $A$ are precisely the worlds where $A$ \textit{could come to be known}.\footnote{One might wonder about the closure conditions on topologies. Finite intersections can perhaps be accounted for by identifying them with sequences of measurements, but what about unions? One intuition comes by observing that for any set $A$, $\int(A) = \bigcup \{U \in \cT \: : \: U \subseteq A\}$, so $\int(A)$ is the information state that arises from learning \textit{that} $A$ is true without learning \textit{what particular measurement} was taken to ascertain this fact. This idea is formalized in \cite{Bjorndahl18} using public announcements; we direct the reader to this work for a more detailed discussion of this point.}

The dual of the interior operator is called \emph{closure}:
\begin{eqnarray*}
\cl(A) & = & X \mysetminus \int(X \mysetminus A)\\
& = & \{x \in X \: : \: (\forall U \in \cT)(x \in U \rimp U \cap A \neq \emptyset)\}.
\end{eqnarray*}
Thus, epistemically speaking, worlds in the closure of $A$ are precisely those worlds in which $A$ is compatible with \textit{every possible measurement}. The closure operator therefore offers a mathematical realization of our intuition about nondeterminism: namely, that a nondeterministic outcome of a program is one that remains possible no matter how good the agent's state of information.

A \defin{topological model} $M$ is a topological space $(X,\cT)$ together with a \emph{valuation} $v: \textsc{prop} \to 2^{X}$. In such models we interpret the basic modal language $\L_{\Box}$ defined by
$$\phi ::= p \, | \, \lnot \phi \, | \, \phi \land \psi \, | \, \Box \phi,$$
where $p \in \textsc{prop}$, via the usual recursive clauses for the Boolean connectives together with the following addition:
\begin{eqnarray*}
x \models \Box \phi & \textrm{ iff } & x \in \int(\val{\phi}),
\end{eqnarray*}
where $\val{\phi} = \{x \in X \: : \: x \models \phi\}$. We also make use of the dual modality $\Diamond$, defined by
\begin{eqnarray*}
x \models \Diamond \phi & \textrm{ iff } & x \in \cl(\val{\phi}).
\end{eqnarray*}
Following the discussion above, we read $\Box \phi$ as ``$\phi$ is knowable'' or ``$\phi$ can be ascertained'' and $\Diamond \phi$ as ``$\phi$ is unfalsifiable'' or ``$\phi$ cannot be ruled out''.

\subsection{Dynamic Topological Models} \label{sec:dtm}

Kremer and Mints \cite{KM05} introduce the notion of a \emph{dynamic topological model}, which is simply a topological model equipped with a continuous function $f: X \to X$. Since we wish to capture the execution of a multitude of programs, we generalize this notion slightly to topological models equipped with a family of functions, one for each program $\pi \in \Pi$. Moreover, continuity is not something we will want to take for granted; we therefore drop this requirement as well.

A \defin{dynamic topological model} is a tuple $(X, \cT, \{f_{\pi}\}_{\pi \in \Pi}, v)$ where $(X,\cT,v)$ is a topological model and each $f_{\pi}: X \to X$. In such models we can interpret the language $\L_{\Box, \scriptsize{\Next}}$ defined by
$$\phi ::= p \, | \, \lnot \phi \, | \, \phi \land \psi \, | \, \Box \phi \, | \, \Next_{\pi} \phi,$$
where $p \in \textsc{prop}$ and $\pi \in \Pi$, via the additional semantic clause:
\begin{eqnarray*}
x \models \Next_\pi \phi & \textrm{ iff } & f_\pi(x) \models \phi.
\end{eqnarray*}
This provides the final tool we need to formalize the re-interpretation of nondeterministic program execution sketched in Section \ref{sec:iam}:
$$\nNext{\pi} \phi \equiv \Diamond \Next_\pi \phi.$$
Semantically:
\begin{eqnarray*}
x \models \nNext{\pi}\phi & \textrm{iff} & x \models \Diamond \Next_{\pi} \phi\\
& \textrm{iff} & x \in \cl(f_{\pi}^{-1}(\val{\phi})).
\end{eqnarray*}
So: $\phi$ is a nondeterministic outcome of $\pi$ (at $x$) just in case it cannot be ruled out (at $x$) that $\phi$ will hold after $\pi$ is executed. Topologically: every measurement at $x$ (i.e., every open neighbourhood of $x$) is compatible with a state where $\phi$ results from executing $\pi$. Call this the \emph{epistemic interpretation of nondeterminism}.

\subsection{Determinism as Continuity}

The epistemic interpretation of nondeterminism accords with our earlier intuitions about random number generators and database queries. Consider a process $rand$ that randomly displays either a $0$ or a $1$, and an agent who (we presume) is unable to measure in advance the relevant quantities that determine the output of this process. This means that both $\Next_{rand}0$ and $\Next_{rand}1$ are compatible with every measurement the agent can take (in advance of running the process), so $\Diamond \Next_{rand}0$ and $\Diamond \Next_{rand}1$ both hold, i.e., $\nNext{rand} 0 \land \nNext{rand} 1$. By contrast, consider a process $query$ that outputs the next entry in a given database and an agent who can look up that entry in advance (which is, say, $0$). This means there is a measurement that guarantees $\Next_{query} 0$, so $\Box \Next_{query} 0$ holds, which yields $[query] 0$.

What exactly is (non)determinism in this setting? It is tempting to describe a (non)deterministic process as one in which the output state can(not) be determined in advance. But this is far too liberal: in principle, states may encode many details that are far beyond the ability of the agent to measure precisely, which would then trivially render every process nondeterministic. For instance, if the state description encodes the current temperature of some internal components of the system (e.g., as a seed value for the $rand$ process), then even \textit{after} executing a simple program like $query$ the user will not be in a position to know the true state.

The correct notion is more subtle. Consider again the $rand$ process: the crucial feature of this program is not that it produces a \textit{state} that cannot be determined in advance, but that it produces a \textit{measurable quantity}---namely, the number displayed---that cannot be determined in advance. 

To make the same point slightly more formally: it may be that no measurement at $x$ rules out \textit{all} the other states (indeed, this will be the case whenever $\{x\}$ is not open). This is necessary but \textit{not sufficient} for nondeterminism because it may still be possible to learn (in advance) everything there is to know about the effects of executing $\pi$ (as describable in the language).

This suggests the following refined account of determinism: a deterministic process is one in which everything \textit{one could learn} about the state of the system \textit{after the program is executed} one can also determine in advance of the program execution. 

This account aligns perfectly with the topological notion of \emph{continuity}. Intuitively: a function is continuous if small changes in the input produce small changes in the output. Topologically: $f$ is \defin{continuous at $x$} if, for every open neighbourhood $V$ of $f(x)$, there exists an open neighbourhood $U$ of $x$ such that $f(U) \subseteq V$. And, finally, epistemically: $f$ is continuous at $x$ if every measurement $V$ compatible with the output state $f(x)$ can be guaranteed in advance by some measurement $U$ compatible with the input state $x$. So the definition of continuity corresponds exactly to our refined account of determinism. In other words: \textit{determinism is continuity in the observation topology}.

Continuity of $f_{\pi}$ can be \emph{defined} in the object language $\L_{\Box, \scriptsize{\Next}}$ by the formula
$$\Next_{\pi} \Box \phi \lthen \Box \Next_{\pi} \phi.\footnote{This claim is made precise and proved in Appendix \ref{app:cha}.}$$
Unsurprisingly, this is precisely the scheme that Kremer and Mints call ``the axiom of continuity'' in their axiomatization of the class of (continuous) dynamic topological models \cite{KM05}.
It reads: ``If, \textit{after} executing $\pi$, $\phi$ is (not only true, but also) measurably true, then it is possible to take a measurement \textit{before} executing $\pi$ that guarantees that $\phi$ will be true after executing $\pi$.'' So this scheme expresses the idea that one need not actually execute $\pi$ in order to determine whatever could be determined after its execution: all the measurable effects of $\pi$ can be determined in advance. Again, this is determinism. Continuity is determinism.

\section{Axiomatization and Model Transformation} \label{sec:amt}

\commentout{
The most basic version of PDL (without any operations on programs) is axiomatized by the axioms and rules of inference given in Table \ref{tbl:pdl0}.
\begin{table}[htp]
\begin{center}
\begin{tabularx}{\textwidth}{>{\hsize=.6\hsize}X>{\hsize=1.3\hsize}X>{\hsize=1.1\hsize}X}
\toprule
(CPL) & all propositional tautologies & Classical propositional logic\\
(K$_{\pi}$) & $[\pi](\phi \lthen \psi) \lthen ([\pi]\phi \lthen [\pi]\psi)$
 & Distribution\\
(MP) & from $\phi$ and $\phi \lthen \psi$ deduce $\psi$ & Modus ponens\\
(Nec$_{\pi}$) & from $\phi$ deduce $[\pi]\phi$ & Necessitation\\
\bottomrule
\end{tabularx}
\end{center}
\caption{Axioms and rules of inference for \textsf{PDL$_{0}$}}\label{tbl:pdl0}
\end{table}%
Call this system \textsf{PDL$_{0}$}.
\begin{theorem} \label{thm:pdl}
\textsf{PDL$_{0}$} is a sound and complete axiomatization of the language $\L_{PDL}$ with respect to the class of all PDL models.\footnote{This can be proved using very standard techniques; see, e.g., \cite{BRV01}.}
\end{theorem}

\subsection{Serial models}
}

We restrict our attention in this section to \emph{serial} PDL models, in which each $R_{\pi}$ is serial: $(\forall x)(\exists y)(xR_{\pi}y)$. Thus, we rule out the possibility of a program $\pi$ producing no output at all in some states (intuitively, ``crashing''), which corresponds to the fact that the functions in dynamic topological models are assumed to be total (i.e., everywhere defined). This allows for a cleaner translation between the two paradigms; in Section \ref{sec:kal} we consider a framework that drops this assumption.


\begin{table}[h]
\caption{Axioms and rules of inference for \textsf{SPDL$_{0}$}}\label{tbl:pdl0}
\begin{tabularx}{\textwidth}{>{\hsize=.6\hsize}X>{\hsize=1.3\hsize}X>{\hsize=1.1\hsize}X}
\toprule
(CPL) & all propositional tautologies & Classical propositional logic\\
(K$_{\pi}$) & $[\pi](\phi \lthen \psi) \lthen ([\pi]\phi \lthen [\pi]\psi)$
 & Distribution\\
(D$_{\pi}$) & $[\pi] \phi \lthen \nNext{\pi}\phi$ & Seriality\\
(MP) & from $\phi$ and $\phi \lthen \psi$ deduce $\psi$ & Modus ponens\\
(Nec$_{\pi}$) & from $\phi$ deduce $[\pi]\phi$ & Necessitation\\
\bottomrule
\end{tabularx}
\end{table}%

The most basic version of serial PDL (without any operations on programs) is axiomatized by the axioms and rules of inference given in Table \ref{tbl:pdl0}. Call this system \textsf{SPDL$_{0}$}.

\begin{theorem} \label{thm:spdl}
\textsf{SPDL$_{0}$} is a sound and complete axiomatization of the language $\L_{PDL}$ with respect to the class of all serial PDL models.\footnote{This can be proved using very standard techniques; see, e.g., \cite{BRV01}.}
\end{theorem}

Using the epistemic interpretation of nondeterminism given in Section \ref{sec:dtm}, we can also interpret the language $\L_{PDL}$ directly in dynamic topological models:
\begin{eqnarray*}
x \models \nNext{\pi}\phi & \textrm{iff} & x \in \cl(f_{\pi}^{-1}(\val{\phi})).
\end{eqnarray*}
And, dually:
\begin{eqnarray*}
x \models [\pi]\phi & \textrm{iff} & x \in \int(f_{\pi}^{-1}(\val{\phi})).
\end{eqnarray*}
This puts us in a position to evaluate our re-interpretation in a precise way. Namely, we can ask: are all the properties of nondeterministic program execution that are captured by standard (serial) PDL models preserved under this new interpretation? And we can answer in the affirmative:
\begin{theorem} \label{thm:spdl2}
\textsf{SPDL$_{0}$} is a sound and complete axiomatization of the language $\L_{PDL}$ with respect to the class of all dynamic topological models.
\end{theorem}

\begin{proof}
Soundness of (CPL) and (MP) is immediate. Soundness of (Nec$_{\pi}$) follows from the fact that $f_{\pi}^{-1}(X) = X$ and $\int(X) = X$, and soundness of (D$_{\pi}$) follows from the fact that, for all $A \subseteq X$, $\int(A) \subseteq \cl(A)$. Finally, to see that (K$_{\pi}$) is sound, observe that
\begin{eqnarray*}
\int(f_{\pi}^{-1}(\val{\phi \lthen \psi})) \cap \int(f_{\pi}^{-1}(\val{\phi})) & = & \int(f_{\pi}^{-1}(\val{\phi \lthen \psi}) \cap f_{\pi}^{-1}(\val{\phi}))\\
& \subseteq & \int(f_{\pi}^{-1}(\val{\psi})).
\end{eqnarray*}
The proof of completeness proceeds by way of a model-transformation construction we provide in Appendix \ref{app:tra}: specifically, we show that every serial PDL model can be transformed into a dynamic topological model in a manner that preserves the truth of all formulas in $\L_{PDL}$ (Proposition \ref{pro:tra}). By Theorem \ref{thm:spdl}, every non-theorem of $\mathsf{SPDL_{0}}$ is refuted on some serial PDL model, so our transformation produces a dynamic topological model that refutes the same formula, thereby establishing completeness. \qed
\end{proof}


Typically one works with richer versions of PDL in which the set of programs $\Pi$ is equipped with one or more operations corresponding, intuitively, to ways of building new programs from old programs. Standard examples include:
\begin{itemize}
\item
\emph{Sequencing}: $\pi_{1};\pi_{2}$ executes $\pi_{1}$ followed immediately by $\pi_{2}$.
\item
\emph{Nondeterministic union}: $\pi_{1} \cup \pi_{2}$ nondeterministically chooses to execute either $\pi_{1}$ or $\pi_{2}$.
\item
\emph{Iteration}: $\pi^{*}$ repeatedly executes $\pi$ some nondeterministic finite number of times.
\end{itemize}
Can we make sense of these operations in our enriched epistemic setting? The latter two transform deterministic programs into nondeterministic programs, and for this reason they are difficult to interpret in a setting where program execution is fundamentally deterministic (i.e., interpreted by functions). We return to discuss this point in Section \ref{sec:con}. Sequencing, on the other hand, does not have this issue; one might guess that it is straightforwardly captured by the condition
$$f_{\pi_{1};\pi_{2}} = f_{\pi_2} \circ  f_{\pi_1}.$$

While function composition certainly seems like the natural way to interpret sequential program execution, there is a wrinkle in the axiomatization.
PDL with sequencing is axiomatized by including the following axiom scheme:
$$\textrm{(Seq)} \qquad \nNext{\pi_{1};\pi_{2}} \phi \liff \nNext{\pi_{1}}\nNext{\pi_{2}} \phi.$$
Interestingly, this scheme is \textit{not} valid in arbitary dynamic topological models. This is because
$$\val{\nNext{\pi_{1};\pi_{2}}\phi} = cl(f_{\pi_{1};\pi_{2}}^{-1}(\val{\phi})) = cl(f_{\pi_{1}}^{-1}(f_{\pi_{2}}^{-1}(\val{\phi}))),$$
whereas
$$\val{\nNext{\pi_{1}}\nNext{\pi_{2}} \phi} = cl(f_{\pi_{1}}^{-1}(cl(f_{\pi_{2}}^{-1}(\val{\phi}))));$$
the extra closure operator means we have
$$\val{\nNext{\pi_{1};\pi_{2}}\phi} \subseteq \val{\nNext{\pi_{1}}\nNext{\pi_{2}} \phi}$$
but not, in general, equality.

A function $f: X \to Y$ is called \defin{open} if for every open subset $U \subseteq X$, the set $f(U)$ is open in $Y$. It turns out that when the function $f_{\pi_{1}}$ is open, the mismatch above vanishes (all of the following claims are proved in Appendix \ref{app:seq}):

\begin{lemma} \label{lem:seq}
Let $(X, \cT, \{f_{\pi}\}_{\pi \in \Pi}, v)$ be a dynamic topological model. If $f_{\pi_{1}}$ is open, then
$$\val{\nNext{\pi_{1};\pi_{2}}\phi} = \val{\nNext{\pi_{1}}\nNext{\pi_{2}} \phi}.$$
\end{lemma}

Say that a dynamic topological model is \emph{open} if each $f_{\pi}$ is open.
\begin{theorem} \label{thm:seq}
\textsf{SPDL$_{0}$} + (Seq) is a sound and complete axiomatization of the language $\L_{PDL}$ with respect to the class of all open dynamic topological models.
\end{theorem}

Like continuity, openness of the function $f_{\pi}$ can be defined in the object language; in fact, it is defined by the converse of the scheme defining continuity:
$$\Box \Next_{\pi} \phi \lthen \Next_{\pi} \Box \phi.$$
Roughly speaking, this says that whatever you can (in principle) predict about executing $\pi$ beforehand you could also come to know afterward. This has a ``perfect recall'' type flavour, except the relevant epistemic notion is not what is actually known but what \textit{could come to be known}. Besides serving to validate the standard sequencing axiom, this principle also plays a crucial role in the next section, where we extend the present framework to incorporate knowledge.

\section{Knowledge and Learning} \label{sec:kal}

To study the epistemology of nondeterministic program execution, we want to be able to reason not only about what \textit{can be} known, but also about what \textit{is} known. To do so we need a richer semantic setting, for which we turn to \emph{topological subset models} \cite{DMP96}; essentially, these use an additional parameter to keep track of the current state of information, through which a standard knowledge modality can be interpreted.

Topological subset models have experienced renewed interest in recent years \cite{Bjorndahl18,vDKO15,BOS17,BO17}, beginning with the work in \cite{Bjorndahl18} studying public announcements in the topological setting. Standard semantics for public announcement logic take the \emph{precondition} of an announcement of $\phi$ to be the truth of $\phi$; in the topological setting, this precondition is strengthened to the \textit{knowability} of $\phi$. As we will see, this interpretation of public announcements is recapitulated in the present framework via a natural interpretation of \emph{test programs}.

\subsection{Incorporating Knowledge}

A \defin{topological subset model} just \emph{is} a topological model $(X,\cT,v)$; the difference lies in the semantic clauses for truth, which are defined with respect to \textit{pairs} of the form $(x,U)$, where $x \in U \in \cT$; such pairs are called \emph{epistemic scenarios}. Intuitively, $x$ represents the actual world, while $U$ captures the agent's current information and thus what they know. Formally, we interpret the language $\L_{K, \Box}$ given by
$$\phi ::= p \, | \, \lnot \phi \, | \, \phi \land \psi \, | \, K \phi \, | \, \Box \phi,$$
where $p \in \textsc{prop}$, as follows:
\begin{eqnarray*}
(x,U) \models p & \textrm{ iff } & x \in v(p)\\
(x,U) \models \lnot \phi & \textrm{ iff } & (x,U) \not\models \phi\\
(x,U) \models \phi \land \psi & \textrm{ iff } & (x,U) \models \phi \textrm{ and } (x,U) \models \psi\\
(x,U) \models K \phi & \textrm{ iff } & U \subseteq \val{\phi}^{U}\\
(x,U) \models \Box \phi & \textrm{ iff } & x \in \int(\val{\phi}^{U}),
\end{eqnarray*}
where $\val{\phi}^{U} = \{x \in U \: : \: (x,U) \models \phi\}$. So the agent knows $\phi$ in the epistemic scenario $(x,U)$ just in case it is guaranteed by their current information $U$.

We next define \emph{dynamic topological subset models} by incorporating functions $f_{\pi}$ as above. Of course, we need subset-style semantics for the dynamic modalities. Perhaps the most natural way to define the updated epistemic scenario is as follows: 
\begin{eqnarray*}
(x,U) \models \Next_{\pi} \phi & \textrm{ iff } & (f_{\pi}(x), f_{\pi}(U)) \models \phi.
\end{eqnarray*}

This definition raises two issues, one technical and the other conceptual. First, as a technical matter, the definition only makes sense if $f_{\pi}(U)$ is open---otherwise $(f_{\pi}(x), f_{\pi}(U))$ is not an epistemic scenario. So we have here another reason to restrict our attention to \emph{open} functions $f_{\pi}$.

Second, conceptually, in a sense this framework does not permit \textit{learning}. True, an agent's state of knowledge changes in accordance with program execution, but every ``live'' possibility $y \in U$ is preserved as the corresponding state $f_{\pi}(y)$ in the updated information set $f_{\pi}(U)$. Intuitively, then, the agent can never truly eliminate possibilities.

\emph{Dynamic epistemic logic} \cite{vDvdHK08} is a modern and vibrant area of research concerned exactly with this issue of how to capture the dynamics of information update. But rather than explicitly importing machinery from this paradigm (e.g., announcements) to represent learning in the present context, we can take advantage of a mechanic that PDL already has available: crashing. In Section \ref{sec:amt} we restricted attention to standard PDL models that were serial, corresponding in our framework to total functions. We now drop this assumption to allow \emph{partial} functions $f_{\pi}: X \pto X$ that are undefined at some points in $X$. This allows the corresponding updates to effectively delete states and thus capture information update in much the same way that, e.g., public announcements do.

A \defin{dynamic topological subset model (over $\Pi$)} is a topological subset model together with a family of partial, open functions\footnote{Typically the concept of openness is applied to total functions, but the definition makes sense for partial functions as well: $f$ is open provided, for all open $U$, $f(U) = \{y \in X \: : \: (\exists x \in U)(f(x) = y)\}$ is open.} $f_{\pi}: X \pto X$, $\pi \in \Pi$. Formulas of the language $\L_{K,\Box,\scriptsize{\Next}}$ given by
$$\phi ::= p \, | \, \lnot \phi \, | \, \phi \land \psi \, | \, K \phi \, | \, \Box \phi \, | \, \Next_{\pi} \phi,$$
are interpreted at epistemic scenarios via the semantic clauses introduced above, taking care to note that the righthand side of the clause defining $\Next_{\pi}$ now carries the implicit assumption that $f_{\pi}(x)$ is actually defined. We provide a sound and complete axiomatization of this logic in Appendix \ref{app:dtel}.


\subsection{Test Programs and Public Announcements} \label{sec:tst}

A standard enrichment of PDL expands the set of programs $\Pi$ to include \emph{test programs}. Unlike other program constructions, test programs are not built from existing programs but instead from formulas in the language: given a formula $\phi$, the program $\phi?$ is introduced to be interpreted by the relation $R_{\phi?}$ defined by
$$xR_{\phi?}y \textrm{ iff } x = y \textrm{ and } x \models \phi.$$
Intuitively, the program $\phi?$ crashes at states where $\phi$ is false, and otherwise does nothing.

Test programs are deterministic, so can be represented just as easily by functions:
$$
f_{\phi?}(x) = \left\{ \begin{array}{ll}
x & \textrm{if $x \models \phi$}\\
\textrm{undefined} & \textrm{otherwise.}
\end{array} \right.
$$
But to make sense of this definition in dynamic topological subset models, two issues must be addressed. First, the relation $x \models \phi$ is not actually defined in subset models---formulas are evaluated with respect to epistemic scenarios, not individual states. However, it is easy to see that when $\phi$ belongs the fragment $\L_{\Box,\scriptsize{\Next}}$, its truth in an epistemic scenario $(x,U)$ is independent of $U$; in this case, we can simply declare that $x \models \phi$ iff $(x,U) \models \phi$ for some (equivalently, all) open sets $U$ containing $x$. We therefore restrict the formation of test programs to $\phi \in \L_{\Box,\scriptsize{\Next}}$. 

Second, $f_{\phi?}$ may not be open. Indeed, $f_{\phi?}$ is open just in case $\val{\phi} = \{x \: : \: x \models \phi\}$ is open. If $\val{\phi}$ is not open then it contains at least one state $x \in \val{\phi} \mysetminus \int(\val{\phi})$, that is, a state at which $\phi$ is true but not \textit{measurably} true. Intuitively, at such states the test program $\phi?$ should crash, since it must fail to determine that $\phi$ is true. This motivates the following revised definition of $f_{\phi?}$:
$$
f_{\phi?}(x) = \left\{ \begin{array}{ll}
x & \textrm{if $x \in \int(\val{\phi})$}\\
\textrm{undefined} & \textrm{otherwise.}
\end{array} \right.
$$
These functions \textit{are} open. Moreover, under these revised semantics, we have:
\begin{eqnarray*}
(x,U) \models \Next_{\phi?}\psi & \textrm{iff} & (f_{\phi?}(x), f_{\phi?}(U)) \models \psi\\
& \textrm{iff} & x \in \int(\val{\phi}) \textrm{ and } (x, U \cap \int(\val{\phi})) \models \psi,
\end{eqnarray*}
which coincides exactly with the topological definition of a public announcement of $\phi$ as given in \cite{Bjorndahl18}.\footnote{Or, rather, it coincides with the dual of the definition given in \cite{Bjorndahl18}, but this is not an important difference.}

\section{Future Work} \label{sec:con}

We have formalized a relatively simple idea: namely, that the nondeterministic outcomes of a process are precisely those that the agent cannot rule out in advance. Using the tools of topology to represent potential observations, we have demonstrated a striking connection between deterministic processes and continuous functions, proved that certain axiomatizations of PDL remain sound and complete when reinterpreted in this enriched setting, and established a natural identity between test programs and public announcements.

Many questions remain, both conceptual and technical. What is the relationship between the (probabilistic) notion of \textit{chance} and the topological construal of nondeterminism presented here? Is there a way to import ``nondeterministic'' operations on programs, such as nondeterministic union or iteration, into this setting? Or is it perhaps better to focus on deterministic analogues of these program constructions, such as, ``If $\phi$ do $\pi_{1}$, else do $\pi_{2}$'', or, ``Do $\pi$ until $\phi$''? How much of dynamic epistemic logic can be recovered as program execution in dynamic topological subset models? For instance, can we make sense of test programs based on epistemic formulas (i.e., formulas that include the $K$ modality), as we can with public announcements? And how can we extend the axiomatization of dynamic topological subset models given in Appendix \ref{app:dtel} to include test programs? These questions and more are the subject of ongoing research.

\bibliographystyle{splncs}
\bibliography{../../../Research/Bibliography/abjorndahl}

\appendix

\section{Proofs and Details} \label{app:prf}

\subsection{Characterizing continuity} \label{app:cha}

A \defin{dynamic topological frame (over $\Pi$)} is a tuple $F = (X, \cT, \{f_{\pi}\}_{\pi \in \Pi})$ where $(X,\cT)$ is a topological space and each $f_{\pi}: X \to X$. In other words, a frame is simply a dynamic topological model without a valuation function. A frame $F$ is said to \emph{validate} a formula $\phi$ just in case $\phi$ is true at every point of every model of the form $(F,v)$.

\begin{proposition}
The formula scheme $\Next_{\pi} \Box \phi \lthen \Box \Next_{\pi} \phi$ defines the class of dynamic topological frames in which $f_{\pi}$ is continuous: that is, for every dynamic topological frame $F$, $F$ validates every instance of $\Next_{\pi} \Box \phi \lthen \Box \Next_{\pi} \phi$ iff $f_{\pi}$ is continuous.
\end{proposition}

\begin{proof}
First suppose that $M$ is a dynamic topological model in which $f_{\pi}$ is continuous, and let $x$ be a point in this model satisfying $\Next_{\pi} \Box \phi$. Then $f_{\pi}(x) \in \int(\val{\phi})$. By continuity, the set $U = f_{\pi}^{-1}(\int(\val{\phi}))$ is open. Moreover, it is easy to see that $x \in U$ and $U \subseteq \val{\Next_{\pi}\phi}$, from which it follows that $x \models \Box \Next_{\pi} \phi$.

Conversely, suppose that $F$ is a dynamic topological frame in which $f_{\pi}$ is not continuous. Let $U$ be an open subset of $X$ such that $A = f_{\pi}^{-1}(U)$ is not open, and let $x \in A \mysetminus \int(A)$; consider a valuation $v$ such that $v(p) = U$. In the resulting model, since $f_{\pi}(x) \in U = \int(U)$, we have $x \models \Next_{\pi} \Box p$. On the other hand, since by definition $x \notin \int(f_{\pi}^{-1}(U))$, we have $x \not\models \Box \Next_{\pi} p$. \qed
\end{proof}

\subsection{Model transformation} \label{app:tra}

Our task in this section is to transform an arbitrary serial PDL model into a dynamic topological model in a truth-preserving manner. The intuition for this transformation is fairly straightforward: in a serial PDL model, each state may be nondeterministically compatible with many possible execution paths corresponding to all the possible ways of successively traversing $R_{\pi}$-edges. In a dynamic topological model, by contrast, all such execution paths must be differentiated by state---roughly speaking, this means we need to create a new state for each possible execution path in the standard model. Then, to preserve the original notion of nondeterminism, we overlay a topological structure that ``remembers'' which new states originated from the same state in the standard model by rendering them topologically indistinguishable.

Let $M = (X, (R_{\pi})_{\pi \in \Pi}, v)$ be a serial PDL model. Let $\Pi^{*}$ denote the set of all finite sequences from $\Pi$. A map $\alpha: \Pi^{*} \to X$ is called a \defin{network (through $M$)} provided $(\forall \vec{\pi} \in \Pi^{*})(\forall \pi \in \Pi)(\alpha(\vec{\pi})R_{\pi}\alpha(\vec{\pi}, \pi))$. In other words, a network $\alpha$ must respect $R_{\pi}$-edges in the sense that it associates with each sequence $(\pi_{1}, \ldots, \pi_{n})$ a path $(x_{1}, \ldots, x_{n+1})$ through $X$ such that, for each $i$, $x_{i}R_{\pi_{i}}x_{i+1}$.\footnote{Specifically, $x_{1} = \alpha(\emptyset)$ and $(\forall i \geq 2)(x_{i} = \alpha(x_{1}, \ldots, x_{i-1}))$.} Networks through $M$ constitute the points of the dynamic topological model we are building:
$$\tilde{X} = \{\alpha \: : \: \textrm{$\alpha$ is a network through $M$.}\}.$$

The topology we equip $\tilde{X}$ with is particularly simple: for each $x \in X$, let $U_{x} = \{\alpha \in \tilde{X} \: : \: \alpha(\emptyset) = x\}$. Clearly the sets $U_{x}$ partition $X$ and so form a topological basis; let $\cT$ be the topology they generate.

Next we define the functions $f_{\pi}: \tilde{X} \to \tilde{X}$. Intuitively, $\alpha \in \tilde{X}$ is a complete record of what paths will be traversed in the original state space $X$ for every sequence of program executions. Therefore, after executing $\pi$, the updated record $f_{\pi}(\alpha)$ should simply consist in those paths determined by $\alpha$ that start with an execution of $\pi$:
$$f_{\pi}(\alpha)(\vec{\pi}) = \alpha(\pi,\vec{\pi}).$$

Finally, define $\tilde{v}: \textsc{prop} \to 2^{\tilde{X}}$ by
$$\tilde{v}(p) = \{\alpha \in \tilde{X} \: : \: \alpha(\emptyset) \in v(p)\}.$$

Let $\tilde{M} = (\tilde{X}, \cT, (f_{\pi})_{\pi \in \Pi}, \tilde{v})$.

\begin{proposition} \label{pro:tra}
For every $\phi \in \L_{PDL}$, for every $\alpha \in \tilde{X}$, $(\tilde{M},\alpha) \models \phi \textrm{ iff } (M,\alpha(\emptyset)) \models \phi$.
\end{proposition}

\begin{proof}
We proceed by induction on the structure of $\phi$. The base case when $\phi = p \in \textsc{prop}$ follows directly from the definition of $\tilde{v}$:
\begin{eqnarray*}
(\tilde{M},\alpha) \models p & \textrm{iff} & \alpha \in \tilde{v}(p)\\
& \textrm{iff} & \alpha(\emptyset) \in v(p)\\
& \textrm{iff} & (M, \alpha(\emptyset)) \models p.
\end{eqnarray*}
The inductive steps for the Boolean connectives are straightforward. So suppose inductively that the result holds for $\phi$; we wish to show it holds for $\nNext{\pi} \phi$.

Let $\alpha \in \tilde{X}$ and $x = \alpha(\emptyset)$. By definition, $(\tilde{M},\alpha) \models \nNext{\pi} \phi$ iff $\alpha \in \cl(f_{\pi}^{-1}(\val{\phi}_{\tilde{M}}))$. Since the topology is generated by a partition, we know that $U_{x}$ is a minimal neighbourhood of $\alpha$, and therefore the preceding condition is equivalent to:
$$U_{x} \cap f_{\pi}^{-1}(\val{\phi}_{\tilde{M}}) \neq \emptyset.$$
This intersection is nonempty just in case there exists an $\alpha' \in \tilde{X}$ such that $\alpha'(\emptyset) = x$ and $f_{\pi}(\alpha') \in \val{\phi}_{\tilde{M}}$. By the induction hypothesis,
\begin{eqnarray*}
f_{\pi}(\alpha') \in \val{\phi}_{\tilde{M}} & \textrm{iff} & (\tilde{M}, f_{\pi}(\alpha')
) \models \phi\\
& \textrm{iff} & (M, f_{\pi}(\alpha')(\emptyset)) \models \phi\\
& \textrm{iff} & (M, \alpha'(\pi)) \models \phi.
\end{eqnarray*}
So to summarize, we have shown that $(\tilde{M},\alpha) \models \nNext{\pi} \phi$ iff there exists an $\alpha' \in \tilde{X}$ such that $\alpha'(\emptyset) = x$ and $(M, \alpha'(\pi)) \models \phi$. Since we know that $\alpha'(\emptyset) R_{\pi} \alpha'(\pi)$, this is in turn equivalent to $(M,x) \models \nNext{\pi} \phi$, which completes the induction. \qed
\end{proof}

\subsection{Sequencing} \label{app:seq}

\begin{customlemma}{\ref{lem:seq}}
Let $(X, \cT, \{f_{\pi}\}_{\pi \in \Pi}, v)$ be a dynamic topological model. If $f_{\pi_{1}}$ is open, then
$$\val{\nNext{\pi_{1};\pi_{2}}\phi} = \val{\nNext{\pi_{1}}\nNext{\pi_{2}} \phi}.$$
\end{customlemma}

\begin{proof}
It suffices to show that $\val{\nNext{\pi_{1};\pi_{2}}\phi} \supseteq \val{\nNext{\pi_{1}}\nNext{\pi_{2}} \phi}$. So let
$$x \in \val{\nNext{\pi_{1}}\nNext{\pi_{2}} \phi} = cl(f_{\pi_{1}}^{-1}(cl(f_{\pi_{2}}^{-1}(\val{\phi}))));$$
then for every open neighbourhood $U$ containing $x$, we know that $U \cap f_{\pi_{1}}^{-1}(cl(f_{\pi_{2}}^{-1}(\val{\phi}))) \neq \emptyset$. This implies that $f_{\pi_{1}}(U) \cap cl(f_{\pi_{2}}^{-1}(\val{\phi})) \neq \emptyset$; since $f_{\pi_{1}}(U)$ is open, it follows that $f_{\pi_{1}}(U) \cap f_{\pi_{2}}^{-1}(\val{\phi}) \neq \emptyset$ as well. This then implies that $U \cap f_{\pi_{1}}^{-1}(f_{\pi_{2}}^{-1}(\val{\phi})) \neq \emptyset$, and therefore
$$x \in \cl(f_{\pi_{1}}^{-1}(f_{\pi_{2}}^{-1}(\val{\phi}))) = \val{\nNext{\pi_{1};\pi_{2}}\phi},$$
as desired. \qed
\end{proof}

Say that a dynamic topological model is \emph{open} if each $f_{\pi}$ is open.
\begin{customthm}{\ref{thm:seq}}
\textsf{SPDL$_{0}$} + (Seq) is a sound and complete axiomatization of the language $\L_{PDL}$ with respect to the class of all open dynamic topological models.
\end{customthm}

\begin{proof}
Lemma \ref{lem:seq} shows that (Seq) is valid in the class of all open dynamic topological models. For completeness, it suffices to observe that the dynamic topological model $\tilde{M}$ constructed in Appendix \ref{app:tra} is itself open: indeed, for each basic open $U_{x}$, we have
\begin{eqnarray*}
f_{\pi}(U_{x}) & = & \{\alpha \in \tilde{X} \: : \: xR_{\pi}\alpha(\emptyset)\}\\
& = & \bigcup_{y \in R(x)} U_{y},
\end{eqnarray*}
which of course is open. \qed
\end{proof}

\begin{proposition}
The formula scheme $\Box \Next_{\pi} \phi \lthen \Next_{\pi} \Box \phi$ defines the class of dynamic topological frames in which $f_{\pi}$ is open: that is, for every dynamic topological frame $F$, $F$ validates every instance of $\Box \Next_{\pi} \phi \lthen \Next_{\pi} \Box \phi$ iff $f_{\pi}$ is open.
\end{proposition}

\begin{proof}
First suppose that $M$ is a dynamic topological model in which $f_{\pi}$ is open, and let $x$ be a point in this model satisfying $\Box \Next_{\pi} \phi$. Then $x \in \int(f_{\pi}^{-1}(\val{\phi}))$. By openness, the set $V = f(\int(f_{\pi}^{-1}(\val{\phi})))$ is open. Moreover, it is easy to see that $f_{\pi}(x) \in V$ and $V \subseteq \val{\phi}$, from which it follows that $x \models \Next_{\pi} \Box \phi$.

Conversely, suppose that $F$ is a dynamic topological frame in which $f_{\pi}$ is not open. Let $U$ be an open subset of $X$ such that $A = f_{\pi}(U)$ is not open, and let $x \in U$ be such that $f_{\pi}(x) \in A \mysetminus \int(A)$; consider a valuation $v$ such that $v(p) = A$. In the resulting model, since $x \in U$ and $f(U) \subseteq \val{p}$, we have $x \models \Box \Next_{\pi} p$. On the other hand, since by definition $f_{\pi}(x) \notin \int(\val{p})$, we have $x \not\models \Next_{\pi} \Box p$. \qed
\end{proof}

\subsection{Dynamic Topological Epistemic Logic} \label{app:dtel}

In this section we provide a sound and complete axiomatization of the language $\L_{K,\Box,\scriptsize{\Next}}$ with respect to the class of all dynamic topological subset models.

Let $\mathsf{CPL}$ denote the axioms and rules of classical propositional logic, let $\mathsf{S5}_{K}$ denote the $\mathsf{S5}$ axioms and rules for the $K$ modality, and let $\mathsf{S4}_{\Box}$ denote the $\mathsf{S4}$ axioms and rules for the $\Box$ modality (see, e.g., \cite{BRV01}). Let $\mathrm{(KI)}$ denote the axiom scheme $K \phi \lthen \Box \phi$, and set
$$\mathsf{EL}_{K,\Box} \defeq \mathsf{CPL} + \mathsf{S5}_{K} + \mathsf{S4}_{\Box} + \mathrm{(KI)}.$$

\begin{theorem}[{\cite[Theorem 1]{Bjorndahl18}}] \label{thm:sac}
$\mathsf{EL}_{K,\Box}$ is a sound and complete axiomatization of $\L_{K,\Box}$ with respect to the class of all dynamic topological subset models.
\end{theorem}

Let $\mathsf{DTEL}$ denote the axiom system $\mathsf{EL}_{K, \Box}$ together with the axiom schemes and rules of inference given in Table \ref{tbl:dtel}.
\begin{table}
\caption{Additional axioms and rules of inference for \textsf{DTEL}}\label{tbl:dtel}
\begin{tabularx}{\textwidth}{>{\hsize=.5\hsize}X>{\hsize=1.5\hsize}X>{\hsize=1.0\hsize}X}
\toprule
($\lnot$-PC$_{\pi}$) & $\Next_{\pi} \lnot \phi \liff (\lnot \Next_{\pi}\phi \land \Next_{\pi} \verum)$ & Partial commutativity of $\lnot$\\
($\land$-C$_{\pi}$) & $\Next_{\pi}(\phi \land \psi) \liff (\Next_{\pi}\phi \land \Next_{\pi}\psi)$ & Commutativity of $\land$\\
($K$-PC$_{\pi}$) & $\Next_{\pi} \verum \lthen (\Next_{\pi} K\phi \liff K(\Next_{\pi} \verum \lthen \Next_{\pi} \phi))$ & Partial commutativity of $K$\\
(O$_{\pi}$) & $(\Box \lnot \Next_{\pi} \phi \land \Next_{\pi} \verum) \lthen \Next_{\pi} \Box \lnot \phi$ & Openness\\
(Mon$_{\pi}$) & from $\phi \lthen \psi$ deduce $\Next_{\pi}\phi \lthen \Next_{\pi}\psi$ & Monotonicity\\
\bottomrule
\end{tabularx}
\end{table}

\begin{theorem}
$\mathsf{DTEL}$ is a sound and complete axiomatization of $\L_{K,\Box,\scriptsize{\Next}}$ with respect to the class of all dynamic topological subset models.
\end{theorem}

\begin{proof}
Soundness of $\mathsf{EL}_{K,\Box}$ follows as in the proof given in \cite[Theorem 1]{Bjorndahl18}, while soundness of the additions presented in Table \ref{tbl:dtel} is easy to check. The presence of $\Next_{\pi}\verum$ in ($\lnot$-PC$_{\pi}$) accounts for the fact that $f_{\pi}$ can be partial (since both $\lnot \Next_{\pi} \phi$ and $\lnot \Next_{\pi} \lnot \phi$ are true at states where $f_{\pi}$ is undefined), and plays an analogous role in ($K$-PC$_{\pi}$) and (O$_{\pi}$). Similarly, the usual ``necessitation'' rule for $\Next_{\pi}$ is not valid, since even if $\phi$ is a theorem, $\Next_{\pi} \phi$ still fails at states where $f_{\pi}$ is undefined.

Completeness is proved by a canonical model construction. Let $X$ denote the set of all maximal $\mathsf{DTEL}$-consistent subsets of $\L_{K,\Box,\scriptsize{\Next}}$. Define a binary relation $\sim$ on $X$ by
$$x \sim y \; \dimp \; (\forall \phi \in \L_{K,\Box,\scriptsize{\Next}})(K \phi \in x \dimp K \phi \in y).$$
Clearly $\sim$ is an equivalence relation; let $[x]$ denote the equivalence class of $x$ under $\sim$. For each $x \in X$, define
$$R(x) = \{y \in X \: : \: (\forall \phi \in \L_{K,\Box,\scriptsize{\Next}})(\Box \phi \in x \rimp \phi \in y)\}.$$
Let $\cB = \{R(x) \: : \: x \in X\}$, and let $\cT$ be the topology generated by $\cB$. It is easy to check that $\cB$ is a basis for $\cT$, and each $R(x)$ is a minimal neighbourhood about $x$ (see, e.g., \cite{AvBB03}).
Given $x \in X$, define
$$
f_{\pi}(x) = \left\{ \begin{array}{ll}
\{\phi \: : \: \Next_{\pi}\phi \in x\} & \; \textrm{if $\Next_{\pi}\verum \in x$}\\
\textrm{undefined} & \; \textrm{otherwise.}
\end{array}
\right.
$$
\shortv{
\vspace{-3mm}
\begin{lemma}
Each $f_{\pi}$ is a partial, open function $X \pto X$.
\end{lemma}

\begin{proof}
Omitted to comply with length restrictions.
\end{proof}
}
\fullv{
The following lemmas establishes that each $f_{\pi}$ is a partial, open function $X \pto X$.

\begin{lemma}
Whenever $f_{\pi}(x)$ is defined, it is maximal $\mathsf{DTEL}$-consistent.
\end{lemma}

\begin{proof}
First suppose for contradiction that $\phi_{1}, \ldots, \phi_{k} \in f_{\pi}(x)$ are inconsistent. Then $\Next_{\pi} \phi_{1}, \ldots, \Next_{\pi}\phi_{k} \in x$, so it follows that $\Next_{\pi} \phi_{1} \land \cdots \land \Next_{\pi}\phi_{k} \in x$, and therefore (using ($\land$-C$_{\pi}$)) $\Next_{\pi}(\phi_{1} \land \cdots \land \phi_{k}) \in x$. Now we know by assumption that
$$\proves_{\mathsf{DTEL}} (\phi_{1} \land \cdots \land \phi_{k}) \lthen \lnot \verum,$$
so by (Mon$_{\pi}$)
$$\proves_{\mathsf{DTEL}} \Next_{\pi}(\phi_{1} \land \cdots \land \phi_{k}) \lthen \Next_{\pi}\lnot \verum,$$
which implies that $\Next_{\pi}\lnot\verum \in x$. Now ($\lnot$-PC$_{\pi}$) implies that $\lnot \Next_{\pi} \verum \land \Next_{\pi} \verum \in x$, contradicting consistency of $x$. This shows that $f_{\pi}(x)$ is consistent.

Next suppose for contradiction that $\phi \notin f_{\pi}(x)$ and $\lnot \phi \notin f_{\pi}(x)$. Then $\Next_{\pi} \phi \notin x$ and $\Next_{\pi} \lnot \phi \notin x$, so $\lnot \Next_{\pi} \phi \in x$ and $\lnot \Next_{\pi} \lnot \phi \in x$. Since $f_{\pi}(x)$ is defined we know that $\Next_{\pi}\verum \in x$, so we have $\lnot \Next_{\pi} \phi \land \Next_{\pi}\verum \in x$ and $\lnot \Next_{\pi} \lnot \phi \land \Next_{\pi}\verum \in x$. Therefore, by ($\lnot$-PC$_{\pi}$), we have $\Next_{\pi} \lnot \phi \in x$ and $\Next_{\pi} \lnot \lnot \phi \in x$. But this implies that $\lnot \phi \in f_{\pi}(x)$ and $\lnot \lnot \phi \in f_{\pi}(x)$, contradicting consistency and thereby establishing maximality. \qed
\end{proof}

\begin{lemma}
For each $x \in X$, $f_{\pi}(R(x)) \in \cT$.
\end{lemma}

\begin{proof}
Let $z \in f_{\pi}(R(x))$ and let $z' \in R(z)$; it suffices to show that $z' \in f_{\pi}(R(x))$. By choice of $z$, there exists $y \in R(x)$ with $z = f_{\pi}(y) = \{\phi \: : \: \Next_{\pi} \phi \in y\}$. We wish to find a $y' \in R(x)$ with $f_{\pi}(y') = z'$. To this end, consider the set
$$\Gamma = \{\psi \: : \: \Box \psi \in x\} \cup \{\Next_{\pi}\chi \: : \: \chi \in z'\}.$$
It is not hard to see that if $y' \supseteq \Gamma$ then $y'$ has the desired properties; by Lindenbaum's lemma, we will therefore be done if we can show that $\Gamma$ is $\mathsf{DTEL}$-consistent.

Suppose not. Then there are $\psi_{1}, \ldots, \psi_{n} \in \{\psi \: : \: \Box \psi \in x\}$ and $\Next_{\pi}\chi_{1}, \ldots, \Next_{\pi}\chi_{m} \in \{\Next_{\pi}\chi \: : \: \chi \in z'\}$ such that
$$\proves_{\mathsf{DTEL}} (\psi_{1} \land \cdots \land \psi_{n}) \lthen \lnot(\Next_{\pi}\chi_{1} \land \cdots \land \Next_{\pi}\chi_{m}),$$
so by ($\land$-C$_{\pi}$),
$$\proves_{\mathsf{DTEL}} (\psi_{1} \land \cdots \land \psi_{n}) \lthen \lnot \Next_{\pi}(\chi_{1} \land \cdots \land \chi_{m}),$$
hence (using $\mathsf{S4}_{\Box}$)
$$\proves_{\mathsf{DTEL}} (\Box \psi_{1} \land \cdots \land \Box \psi_{n}) \lthen \Box \lnot \Next_{\pi}(\chi_{1} \land \cdots \land \chi_{m}).$$
Now each $\Box \psi_{i}$ is contained in $x$, and therefore also in $y$; it follows that $\Box \lnot \Next_{\pi}(\chi_{1} \land \cdots \land \chi_{m}) \in y$ as well; moreover, since $f_{\pi}(y) = z$, we have $\Next_{\pi}\verum \in y$, so by (O$_{\pi}$) we obtain $\Next_{\pi}\Box \lnot (\chi_{1} \land \cdots \land \chi_{m}) \in y$ and therefore $\Box \lnot (\chi_{1} \land \cdots \land \chi_{m}) \in z$, hence $\lnot (\chi_{1} \land \cdots \land \chi_{m}) \in z'$, a contradiction. \qed
\end{proof}
}

For each $p \in \textsc{prop}$, set
$v(p) \defeq \{x \in X \: : \: p \in x\}$.
Let $\X = (X, \cT, \{f_{\pi}\}_{\pi \in \Pi}, v)$. Clearly $\X$ is a dynamic topological subset model.

\begin{lemma}[Truth Lemma] \label{lem:tru}
For every $\phi \in \L_{K,\Box,\scriptsize{\Next}}$, for all $x \in X$, $\phi \in x$ iff $(\X, x, [x]) \models \phi$.
\end{lemma}

\shortv{
\begin{proof}
Omitted to comply with length restrictions.
\end{proof}
}

\fullv{
\begin{proof}
First we show that $y \in R(x) \rimp y \sim x$. Let $y \in R(x)$. If $K\psi \in x$ then $KK\psi \in x$, so $\Box K \psi \in x$ and therefore $K\psi \in y$; conversely, if $\lnot K\psi \in x$ then $K\lnot K\psi \in x$, so $\Box \lnot K\psi \in x$ and therefore $\lnot K\psi \in y$. Note that this implies that $[x] \in \cT$, so $(x, [x])$ is indeed an epistemic scenario.

The proof proceeds by induction on the structure of $\phi$. The base case holds by definition of $v$, and the inductive steps for the Boolean connectives are straightforward. The inductive step for $K$ mirrors exactly the corresponding step in \cite[Theorem 1]{Bjorndahl18}.

So suppose inductively the result holds for $\phi$ and let us show it holds for $\Box \phi$. If $\Box \phi \in x$ then by definition of $R$ we know that for every $y \in R(x)$, $\phi \in y$. By the inductive hypothesis, this implies that $(\forall y \in R(x))(y,[y]) \models \phi$; since $y \in R(x) \rimp y \sim x$, this is equivalent to $(\forall y \in R(x))((y,[x]) \models \phi)$; since $R(x)$ is an open neighbourhood of $x$, this yields $(x,[x]) \models \Box \phi$.

For the converse, suppose that $\Box \phi \notin x$. Then
$$\{\psi \: : \: \Box \psi \in x\} \cup \{\lnot \phi\}$$
is consistent, for if not there are $\psi_{1}, \ldots, \psi_{m} \in \{\psi \: : \: \Box \psi \in x\}$ such that
$$\proves_{\mathsf{DTEL}} \psi_{1} \land \cdots \land \psi_{m} \lthen \phi,$$
from which it follows that
$$\proves_{\mathsf{DTEL}} \Box\psi_{1} \land \cdots \land \Box\psi_{m} \lthen \Box \phi,$$
which implies $\Box \phi \in x$, a contradiction. By Lindenbaum's lemma, we therefore obtain a point $y \in X$ with $y \in R(x)$ and $\phi \notin y$. This latter fact, by the inductive hypothesis, yields $(y,[y]) \not\models \phi$ and thus $(y,[x]) \not \models \phi$ (since $y \sim x$), which in turn yields $(x,[x]) \not\models \Box\phi$ since $R(x)$ is a minimal neighbourhood of $x$.

For the last inductive step a lemma will be useful.

\begin{lemma} \label{lem:kpc}
For all $x \in X$, if $f_{\pi}(x)$ is defined then $f_{\pi}([x]) = [f_{\pi}(x)]$.
\end{lemma}

\begin{proof}
First suppose that $z \in f_{\pi}([x])$, so there is some $y \sim x$ such that $f_{\pi}(y) = z$. We need to show that $z \sim f_{\pi}(x)$. If $K\psi \in z$ then so is $KK\psi$, so $\Next_{\pi}KK\psi \in y$; since $f_{\pi}(y)$ is defined, certainly $\Next_{\pi}\verum \in y$, and therefore by ($K$-PC$_{\pi}$) we can deduce that $K(\Next_{\pi}\verum \lthen \Next_{\pi}K\psi) \in y$. Since $y \sim x$, this implies that $\Next_{\pi}\verum \lthen \Next_{\pi}K\psi \in x$, and therefore (since $f_{\pi}(x)$ is defined), $\Next_{\pi}K\psi \in x$, hence $K\psi \in f_{\pi}(x)$. An analogous argument shows that if $\lnot K\psi \in z$ then $\lnot K \psi \in f_{\pi}(x)$, which establishes that $z \sim f_{\pi}(x)$.

Conversely, suppose that $z \sim f_{\pi}(x)$; we wish to show that there exists a $y \in [x]$ such that $f_{\pi}(y) = z$. Using Lindenbaum's lemma, we will be done if we can show that the set
$$\{K\psi \: : \: K \psi \in x\} \cup \{\Next_{\pi}\chi \: : \: \chi \in z\}$$
is $\mathsf{DTEL}$-consistent. So suppose not; then there are $K\psi_{1}, \ldots, K\psi_{n} \in \{K\psi \: : \: K \psi \in x\}$ and $\Next_{\pi}\chi_{1}, \ldots, \Next_{\pi}\chi_{m}, \Next_{\pi}\verum \in \{\Next_{\pi}\chi \: : \: \chi \in z\}$ such that
$$\proves_{\mathsf{DTEL}} (K\psi_{1} \land \cdots \land K\psi_{n}) \lthen (\Next_{\pi}\verum \lthen \lnot(\Next_{\pi}\chi_{1} \land \cdots \land \Next_{\pi}\chi_{m})),$$
from which it follows (using $\mathsf{S5}_{K}$) that
$$\proves_{\mathsf{DTEL}} (K\psi_{1} \land \cdots \land K\psi_{n}) \lthen K(\Next_{\pi}\verum \lthen \lnot\Next_{\pi}(\chi_{1} \land \cdots \land \chi_{m})).$$
Observe that $\Next_{\pi}\verum \lthen \lnot \Next_{\pi} \chi$ is propositionally equivalent to $\Next_{\pi}\verum \lthen (\lnot \Next_{\pi} \chi \land \Next_{\pi}\verum)$, and by ($\lnot$-PC$_{\pi}$) this is in turn equivalent in $\mathsf{DTEL}$ to $\Next_{\pi}\verum \lthen \Next_{\pi} \lnot \chi$. Hence
$$\proves_{\mathsf{DTEL}} (K\psi_{1} \land \cdots \land K\psi_{n}) \lthen K(\Next_{\pi}\verum \lthen \Next_{\pi} \lnot (\chi_{1} \land \cdots \land \chi_{m})),$$
so since $K\psi_{1}, \ldots, K\psi_{n} \in x$, it follows that $K(\Next_{\pi}\verum \lthen \Next_{\pi} \lnot (\chi_{1} \land \cdots \land \chi_{m})) \in x$. Since $f_{\pi}(x)$ is defined, also $\Next_{\pi}\verum \in x$, so by ($K$-PC$_{\pi}$) we can deduce that $\Next_{\pi} K \lnot (\chi_{1} \land \cdots \land \chi_{m}) \in x$. We therefore have $K \lnot (\chi_{1} \land \cdots \land \chi_{m}) \in f_{\pi}(x)$, so since $z \sim f_{\pi}(x)$, $\lnot (\chi_{1} \land \cdots \land \chi_{m}) \in z$, a contradiction.
\qed
\end{proof}

Back to the Truth Lemma: suppose inductively the result holds for $\phi$ and let us show it holds for $\Next_{\pi} \phi$. Observe that $\Next_{\pi} \phi \in x$ iff $f_{\pi}(x)$ is defined and $\phi \in f_{\pi}(x)$. By the inductive hypothesis, this is equivalent to $(f_{\pi}(x), [f_{\pi}(x)]) \models \phi$; by Lemma \ref{lem:kpc}, this is in turn equivalent to $(f_{\pi}(x), f_{\pi}[x]) \models \phi$, which by definition holds iff $(x,[x]) \models \Next_{\pi} \phi$.
\qed
\end{proof}
}

Completeness is an easy consequence: if $\phi$ is not a theorem of $\mathsf{DTEL}$, then $\{\lnot \phi\}$ is consistent and so can be extended by Lindenbaum's lemma to some $x \in X$; by Lemma \ref{lem:tru}, we have $(\X,x,[x]) \not\models \phi$.
\qed
\end{proof}

\commentout{
\subsection{Partial functions}

\begin{customthm}{\ref{thm:pdtm}}
\textsf{PDL$_{0}$} is a sound and complete axiomatization of the language $\L_{PDL}$ with respect to the class of all partial dynamic topological models.
\end{customthm}

\begin{proof}
Soundness of (CPL) and (MP) is immediate. Soundness of (Nec$_{\pi}$) follows from the fact that if $\val{\phi} = X$, then
\begin{eqnarray*}
x \models [\pi]\phi & \textrm{iff} & x \in \int(f_{\pi}^{-1}(X) \cup C_{\pi})\\
& \textrm{iff} & x \in \int((X \mysetminus C_{\pi}) \cup C_{\pi})\\
& \textrm{iff} & x \in \int(X)\\
& \textrm{iff} & x \in X.
\end{eqnarray*}
Lastly, to see that (K$_{\pi}$) is sound, observe that
\begin{eqnarray*}
\val{[\pi](\phi \lthen \psi)} \cap \val{[\pi]\phi} & = & \int(f_{\pi}^{-1}(\val{\phi \lthen \psi}) \cup C_{\pi}) \cap \int(f_{\pi}^{-1}(\val{\phi}) \cup C_{\pi})\\
& = & \int((f_{\pi}^{-1}(\val{\phi \lthen \psi}) \cap f_{\pi}^{-1}(\val{\phi})) \cup C_{\pi})\\
& \subseteq & \int(f_{\pi}^{-1}(\val{\psi}) \cup C_{\pi})\\
& = & \val{[\pi]\psi}.
\end{eqnarray*}

The proof of completeness follows exactly the same strategy as before; we simply need to extend our previous model transformation result to show that every PDL model can be transformed into a \textit{partial} dynamic topological model in a way that preserves the truth of all formulas in $\L_{PDL}$. And this is easy: 

\draft{to do}

\end{proof}
}

\end{document}